\documentclass[sigconf]{aamas}

\usepackage{balance} 
\usepackage{algorithm}
\usepackage{algorithmic}
\usepackage{amsmath}

\doi{}

\makeatletter
\gdef\@copyrightpermission{
  \begin{minipage}{0.2\columnwidth}
   \href{https://creativecommons.org/licenses/by/4.0/}{\includegraphics[width=0.90\textwidth]{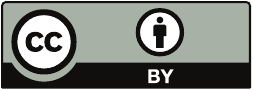}}
  \end{minipage}\hfill
  \begin{minipage}{0.8\columnwidth}
   \href{https://creativecommons.org/licenses/by/4.0/}{This work is licensed under a Creative Commons Attribution International 4.0 License.}
  \end{minipage}
  \vspace{5pt}
}
\makeatother

\setcopyright{ifaamas}
\acmConference[AAMAS '26]{Proc.\@ of the 25th International Conference
on Autonomous Agents and Multiagent Systems (AAMAS 2026)}{May 25 -- 29, 2026}
{Paphos, Cyprus}{C.~Amato, L.~Dennis, V.~Mascardi, J.~Thangarajah (eds.)}
\copyrightyear{2026}
\acmYear{2026}
\acmDOI{}
\acmPrice{}
\acmISBN{}

\title[Sample-Efficient PSRO with JBR]{Sample-Efficient Policy Space Response Oracles with Joint Experience Best Response}

\author{Ariyan Bighashdel}
\affiliation{
  \institution{Utrecht University}
  \city{Utrecht}
  \country{The Netherlands}}
\affiliation{
  \institution{Delft University of Technology}
  \city{Delft}
  \country{The Netherlands}}
\email{a.bighashdel@uu.nl}

\author{Thiago D. Simão}
\affiliation{
  \institution{Eindhoven University of Technology}
  \city{Eindhoven}
  \country{The Netherlands}}
\email{t.simao@tue.nl}

\author{Frans A. Oliehoek}
\affiliation{
  \institution{Delft University of Technology}
  \city{Delft}
  \country{The Netherlands}}
\email{f.a.oliehoek@tudelft.nl}

\begin{abstract}
Multi-agent reinforcement learning (MARL) offers a scalable alternative to exact game-theoretic analysis but suffers from non-stationarity and the need to maintain diverse populations of strategies that capture non-transitive interactions. Policy Space Response Oracles (PSRO) address these issues by iteratively expanding a restricted game with approximate best responses (BRs), yet per-agent BR training makes it prohibitively expensive in many-agent or simulator-expensive settings. We introduce \emph{Joint Experience Best Response (JBR)}, a drop-in modification to PSRO that collects trajectories once under the current meta-strategy profile and reuses this joint dataset to compute BRs for all agents simultaneously. This amortizes environment interaction and improves the \emph{sample efficiency of best-response computation}. Because JBR converts BR computation into an offline RL problem, we propose three remedies for distribution-shift bias: (i) \emph{Conservative JBR} with safe policy improvement, (ii) \emph{Exploration-Augmented JBR} that perturbs data collection and admits theoretical guarantees, and (iii) \emph{Hybrid BR} that interleaves JBR with periodic independent BR updates. Across benchmark multi-agent environments, Exploration-Augmented JBR achieves the best accuracy–efficiency trade-off, while Hybrid BR attains near-PSRO performance at a fraction of the sample cost. Overall, JBR makes PSRO substantially more practical for large-scale strategic learning while preserving equilibrium robustness.
\end{abstract}


\newcommand{\BibTeX}{\rm B\kern-.05em{\sc i\kern-.025em b}\kern-.08em\TeX}
\newcommand{\bri}{\mathrm{BR}_i}

\begin{document}


\pagestyle{fancy}
\fancyhead{}


\maketitle 

\section{Introduction}

\begin{figure}[t]
    \centering
    \includegraphics[width=0.9\linewidth]{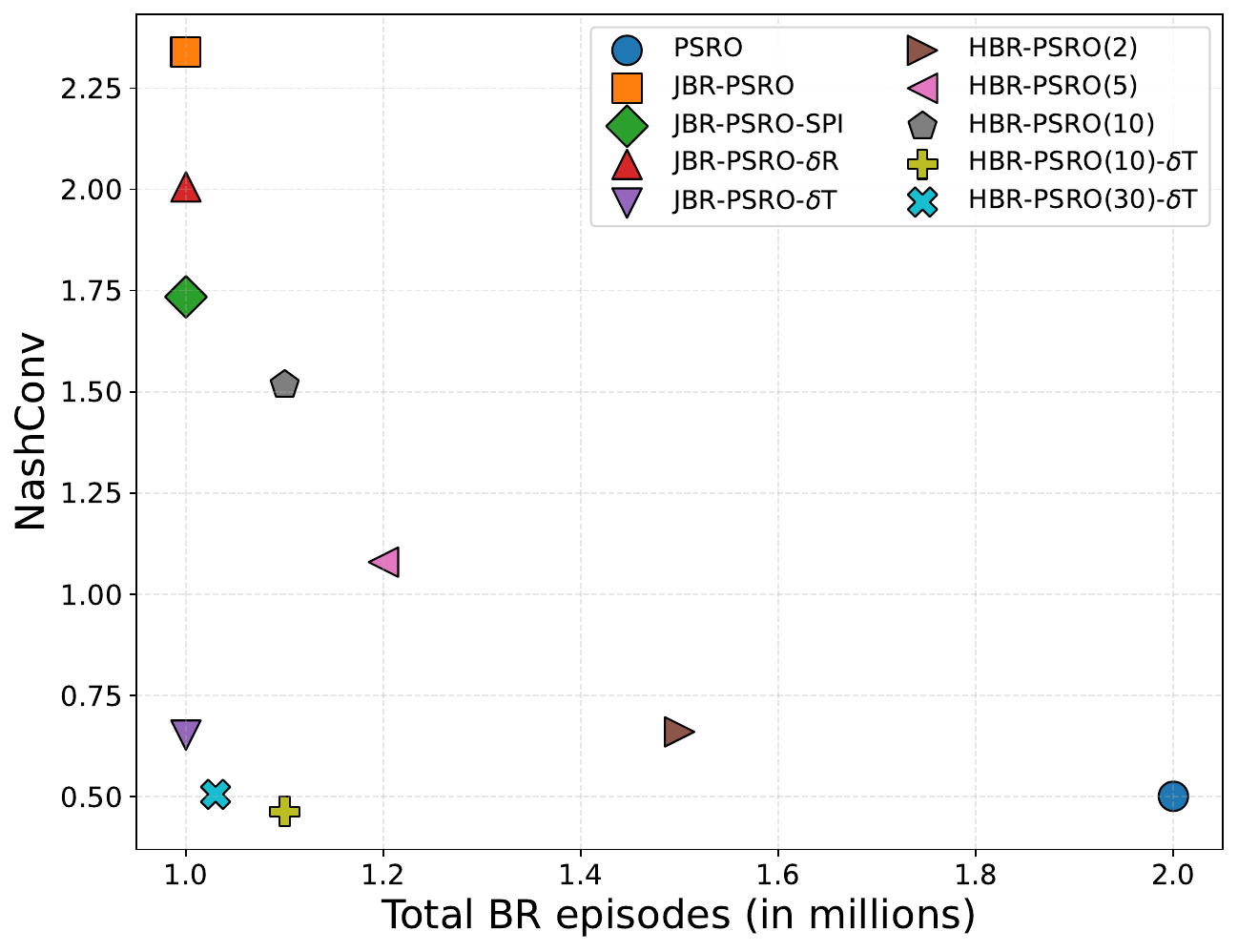}
    \caption{
    \textbf{Sample-efficiency -- accuracy trade-off in Leduc Poker.} 
    Shown are total best-response episodes (in millions, $x$-axis) 
    versus minimum NashConv after 100 iterations ($y$-axis). 
    Standard PSRO is accurate but requires the highest BR sample cost. 
    Joint Experience Best Response (\textsc{JBR-PSRO}) and its enhanced variants—
    conservative (\textsc{JBR-PSRO-SPI}), exploration-augmented 
    (\textsc{JBR-PSRO-$\delta$R}, \textsc{JBR-PSRO-$\delta$T}), and hybrid 
    (\textsc{HBR-PSRO(10/30)-$\delta$T})—drastically reduce the number of 
    BR episodes needed for convergence, with \textsc{JBR-PSRO-$\delta$T} achieving 
    the best trade-off and hybrid versions approaching PSRO-level accuracy. 
    }
    \label{fig:efficiency_tradeoff}
    \Description{
    \textbf{Sample-efficiency -- accuracy trade-off in Leduc Poker.} 
    Shown are total best-response episodes (in millions, $x$-axis) 
    versus minimum NashConv after 100 iterations ($y$-axis). 
    Standard PSRO is accurate but requires the highest BR sample cost. 
    Joint Experience Best Response (\textsc{JBR-PSRO}) and its enhanced variants—
    conservative (\textsc{JBR-PSRO-SPI}), exploration-augmented 
    (\textsc{JBR-PSRO-$\delta$R}, \textsc{JBR-PSRO-$\delta$T}), and hybrid 
    (\textsc{HBR-PSRO(10/30)-$\delta$T})—drastically reduce the number of 
    BR episodes needed for convergence, with \textsc{JBR-PSRO-$\delta$T} achieving 
    the best trade-off and hybrid versions approaching PSRO-level accuracy. 
    }
\end{figure}

Multi-agent systems (MAS) underlie many real-world domains, from autonomous traffic control to electronic markets, where agents must continually adapt to one another’s strategies. Game theory offers strong foundations for analyzing such interactions through equilibrium concepts, but computing equilibria in large or sequential games is infeasible in practice due to the combinatorial growth of policy spaces~\cite{von2002computing,gambit23}. This motivates the need for scalable learning-based approaches.

Multi-agent reinforcement learning (MARL) provides one such alternative, enabling agents to improve their policies through repeated interaction~\cite{shoham2007if,marl-book}. However, MARL faces two key challenges. First, independent learning induces non-stationarity: as each agent updates, the environment faced by others keeps changing~\cite{tuyls2012multiagent}. Centralized Training with Decentralized Execution (CTDE)~\cite{kraemer2016multi,lowe2017multi} partly mitigates this issue by allowing agents to learn with access to joint observations or gradients during training while acting independently at test time. However, residual instability often persists. Second, many games are non-transitive—no strategy is globally dominant—and robustness requires maintaining diverse populations of strategies~\cite{czarnecki2020real,sanjaya2022measuring,bighashdel2024off,bighashdel2023coordinating}.

To address these challenges, the \emph{Policy Space Response Oracles (PSRO)} framework~\cite{lanctot2017unified} emerged as a synthesis of ideas from several research communities: from \emph{planning}, it builds on the Double Oracle algorithm~\cite{mcmahan2003planning}; from \emph{co-evolution}, it adopts mechanisms such as Nash memory~\cite{ficici2003game} and Parallel Nash Memory~\cite{oliehoek2006parallel}; from \emph{empirical game-theoretic analysis (EGTA)}, it incorporates simulation-based restricted games~\cite{schvartzman2009exploring}; and from \emph{reinforcement learning}, it leverages RL-based oracles to train best responses in high-dimensional sequential environments. At its core, PSRO maintains a restricted game consisting of a subset of strategies, computes a \emph{meta-strategy} (a probability distribution over the restricted strategies) using a chosen solver, and expands the restricted game by adding new (approximate) best responses, i.e., strategies that maximize an agent’s payoff against the current meta-strategy profile (the mixed distribution of the other agents’ strategies). By repeating this process, PSRO incrementally approximates equilibria while preserving diverse strategy populations that can capture non-transitive dynamics. The overall procedure is summarized in Algorithm~\ref{alg:psro}.

Despite its strengths, PSRO suffers from a major limitation: its computational cost. At each iteration, every agent must independently compute a best response, requiring as many reinforcement learning problems as there are agents, and in each case, the remaining agents must be simulated. In domains with many players or costly simulators, this process quickly becomes impractical. Existing works have sought to improve efficiency through parallelization~\cite{lanctot2017unified,mcaleer2020pipeline}, sample-efficient simulation~\cite{smith2023co,zhou2022efficient}, or transfer learning~\cite{liu2022neupl,liu2024neural,smith2023strategic}. While such methods reduce wall-clock training time, they do not alter the fundamental structure of PSRO: separate best-response computation for each agent at every iteration.  
 
We propose \emph{Joint Experience Best Response (JBR)}, a drop-in modification to PSRO that reuses joint experience across agents. Instead of training each best response independently, JBR collects trajectories once under the current meta-strategy profile and computes best responses for all agents from this shared dataset. This dramatically improves the \emph{sample efficiency of best-response computation}, making PSRO practical in multi-agent or simulator-expensive domains.

Because JBR reuses data generated under the current meta-strategy rather than interacting with the environment online, best-response computation becomes an \emph{offline reinforcement learning (offline RL)} problem, where policies must be optimized from a fixed dataset instead of new environment rollouts. This offline setting can introduce bias due to distribution shift. We introduce three complementary variants:
\emph{(i) Conservative JBR}, which enforces Safe Policy Improvement~\cite{laroche2019safe} to guarantee no degradation below the baseline meta-strategy;
\emph{(ii) Exploration-Augmented JBR}, which perturbs data collection to improve coverage while preserving theoretical guarantees; and
\emph{(iii) Hybrid BR}, which alternates between JBR and periodic independent BR to combine efficiency with accuracy (Figure~\ref{fig:efficiency_tradeoff}).

Our findings show that exploration-augmented JBR with targeted perturbation (\textbf{JBR-PSRO-$\delta$T}) achieves near-PSRO equilibrium accuracy at a fraction of the sample cost, while hybrid variants (\textbf{HBR-PSRO($k$)-$\delta$T}) reach PSRO-level performance with only minor additional training. Across discrete and continuous environments, PSRO-style methods—including our JBR variants—consistently outperform standard MARL paradigms such as independent learning and CTDE, confirming that strategically grounded training remains more robust when made sample-efficient.

Our contributions are threefold:
\vspace{-2pt}
\begin{enumerate}
    \item We introduce \emph{Joint Experience Best Response (JBR)}, an efficient modification of PSRO that \textbf{improves the sample efficiency of best-response computation} by reusing shared trajectories across agents.  
    \item We identify and address the offline RL limitations of JBR, proposing and evaluating three remedies: Conservative JBR, Exploration-Augmented JBR (with theoretical guarantees), and Hybrid BR.  
    \item We provide a comprehensive empirical study comparing JBR variants not only against Independent BR (standard PSRO), but also against mainstream MARL baselines such as Independent Learning and CTDE, showing that PSRO-based approaches consistently yield superior strategic robustness.  
\end{enumerate}
\vspace{-4pt}

\section{Related Work} 
A wide range of extensions to PSRO have been proposed to improve different aspects of the framework, including scalability, exploration, and equilibrium computation. 
These directions are reviewed in a recent survey~\cite{bighashdel2024policy}. 
In this paper, we focus specifically on methods aimed at improving the \emph{efficiency} of PSRO.

\emph{Parallelization.}  
Early work explored reducing wall-clock time through parallel training of best responses. 
The Deep Cognitive Hierarchy (DCH)~\cite{lanctot2017unified} introduced a hierarchical schedule in which each player trains against equilibria of lower levels, effectively reusing earlier computations. 
Pipeline PSRO (P2SRO)~\cite{mcaleer2020pipeline} refined this idea by coordinating multiple strategies trained in parallel and dynamically optimizing the hierarchy itself. While these methods improve throughput, they require substantial computational resources and do not reduce the overall number of best responses.

\emph{Sample efficiency.}  
Another line of research targets the reliance on costly simulation. 
Smith and Wellman~\cite{smith2023co} proposed to jointly learn a generative model of the game while running PSRO, enabling queries to the learned model instead of the environment. 
Zhou~et~al.~\cite{zhou2022efficient} developed EPSRO, which avoids redundant restricted-game evaluations by exploiting cases where best-response target profiles can be derived without full simulation. Although these approaches cut down on simulator calls, their benefits depend heavily on the quality of the learned model, and the cost of retraining best responses remains unchanged.

\emph{Transfer learning.}  
A third approach is to accelerate adaptation by reusing knowledge across best responses. 
NeuPL~\cite{liu2022neupl} introduced parameter sharing across strategies in a single network, while Liu~et~al.~\cite{liu2022simplex} trained against sampled mixtures to provide robustness to diverse opponent profiles. 
Later, Liu~et~al.~\cite{liu2024neural} combined NeuPL with JPSRO to extend transfer benefits to correlated-equilibrium computation. 
Smith~et~al.~\cite{smith2023strategic} proposed a Mixed-Oracle method that constructs new best responses by mixing value functions pre-trained against pure strategies. These methods accelerate learning through knowledge reuse, but their effectiveness is tied to how well information transfers across opponents, and they still instantiate new best responses one at a time.

\emph{Summary.}  
Prior work demonstrates that PSRO can be made more practical through parallelization, efficient simulation, and transfer across best-response computations. 
While such methods reduce wall-clock training time, they largely remain implementation-level optimizations: the core PSRO loop continues to generate a separate best response for each agent at every iteration. 
In contrast, our proposed JBR framework changes this structure directly by reusing shared interaction experience to compute best responses for all agents simultaneously, offering a different route toward scalability.

\section{Background}
The core idea of PSRO is to approximate otherwise intractable games by iteratively constructing a smaller \emph{restricted game}, 
in which each player has access only to a subset of strategies. 
This restricted game is gradually expanded by adding new best responses, which are computed using RL. 
In this way, PSRO maintains a tractable approximation of the full game while progressively refining the strategic coverage. Formally, consider an $n$-player Markov (stochastic) game defined by:
\begin{itemize}
\item a state space $\mathcal{S}$,
\item action spaces $\mathcal{A}_i$ for each player $i \in N=\{1,\dots,n\}$, with joint action space $\mathcal{A} = \times_{i=1}^n \mathcal{A}_i$,
\item a transition kernel $P(s' \mid s,a_1,\dots,a_n)$,
\item player-specific reward functions $r_i(s,a_1,\dots,a_n)$,
\item a discount factor $\gamma \in [0,1)$,
\item and an initial state distribution $\rho_0$.
\end{itemize}

A strategy for player $i$ is a policy $\pi_i \in \Pi_i$, i.e., a mapping from states to distributions over $\mathcal{A}_i$. 
Given a policy profile $\pi = (\pi_1,\dots,\pi_n)$, the utility of player $i$ is the expected discounted return of per-step rewards:
\[
u_i(\pi) = \mathbb{E}\!\left[\sum_{t=0}^\infty \gamma^t\, r_i(s_t,a_{1,t},\dots,a_{n,t}) \;\middle|\; \pi\right].
\]

PSRO can thus be understood as operating on a \emph{normal-form game induced by the Markov game}, $\mathcal{G} = (N, (\Pi_i), (u_i))$, where the pure strategies are full policies and payoffs are their expected returns. 
In practice, this induced game is intractable because each $\Pi_i$ is enormous. 
PSRO therefore maintains restricted policy sets $X_i \subseteq \Pi_i$, consisting of the policies generated so far, 
and defines the restricted empirical game, $\hat{\mathcal{G}} = (N, (X_i), (\hat{u}_i))$, where $\hat{u}_i$ are payoff estimates obtained from rollouts of profiles $\pi \in X = \times_i X_i$. 
The restricted empirical game serves as the basis for strategy updates, 
representing a smaller, simulation-based approximation of the full game that captures interactions among the currently known policies. The PSRO algorithm alternates between two key components:
\begin{itemize}
\item \textbf{Meta-Strategy Solver (MSS):} Given $\hat{\mathcal{G}}$, the MSS computes a meta-strategy profile (mixed strategy) $\sigma = (\sigma_i)_{i \in N}$, where each $\sigma_i$ is a probability distribution over $X_i$ that specifies how frequently each policy in the agent’s current set is played. Typical choices include Nash equilibrium solvers or regret-minimization dynamics.

\item \textbf{Response Oracle (RO):} For each player $i$, the RO computes an approximate best-response policy $\pi_i^{\mathrm{BR}}$ against $\sigma_{-i}$. Formally, the best-response operator is
\[
    br_i(\sigma_{-i}) = \arg\max_{\pi_i \in \Pi_i} u_i(\pi_i, \sigma_{-i}).
\]
In practice, the RO is \emph{implemented by deep RL}: the agent is trained in the induced MDP where opponents follow $\sigma_{-i}$, yielding a policy $\pi_i^{\mathrm{BR}}$ that approximates $br_i(\sigma_{-i})$. The new policy is then added to $X_i$, expanding the restricted game.
\end{itemize}

\begin{algorithm}[t]
\caption{Policy Space Response Oracles (PSRO)}
\label{alg:psro}
\begin{algorithmic}[1]
\REQUIRE Initial restricted policy sets $(X_i)$
\STATE Estimate $\hat{\mathcal{G}} = (N,(X_i),(\hat{u}_i))$
\STATE Initialize $\sigma \leftarrow \text{MSS}(\hat{\mathcal{G}})$
\WHILE{not terminated}
    \FOR{player $i \in N$}
        \STATE Compute $\pi_i^{\mathrm{BR}} \leftarrow \text{RO}(\sigma_{-i})$
        \STATE Update $X_i \leftarrow X_i \cup \{\pi_i^{\mathrm{BR}}\}$
    \ENDFOR
    \STATE Re-estimate $\hat{\mathcal{G}}$
    \STATE Update $\sigma \leftarrow \text{MSS}(\hat{\mathcal{G}})$
\ENDWHILE
\RETURN $\sigma$
\end{algorithmic}
\end{algorithm}

The complete procedure is shown in Algorithm~\ref{alg:psro}. The process continues until a stopping criterion is met, often based on \emph{regret}. 
The regret of player $i$ under a meta-strategy profile $\sigma$ is
\begin{equation}
    \rho^{\mathcal{G}}_i(\sigma) = \max_{\pi_i \in \Pi_i} u_i(\pi_i, \sigma_{-i}) - u_i(\sigma_i, \sigma_{-i}),
\end{equation}
which measures the potential gain from deviating unilaterally. 
At a Nash equilibrium, all regrets vanish. 
The aggregate regret,
\begin{equation}
    \rho^{\mathcal{G}}(\sigma) = \sum_{i \in N} \rho^{\mathcal{G}}_i(\sigma),
\end{equation}
is known as \emph{NashConv}~\cite{lanctot2017unified} and, in two-player games, corresponds to exploitability. 
This quantity is widely used to track convergence in PSRO experiments. 
In practice, it can only be computed exactly in small discrete games where exhaustive best-response search is feasible; in larger or continuous environments, it must be approximated using learned best responses (see Section~\ref{sec:experiments}).

\section{Method}
\label{sec:method}

Throughout this section we adopt a unified perspective: best responses are derived via \emph{model-based reinforcement learning}, specifically value iteration. 
This choice is made purely for clarity in explaining the different approaches. 
In practice, the RO is typically realized by deep RL methods, which approximate the same objective when exact model-based computation is infeasible. 
All formulations are therefore presented in this model-based framework for consistency, while our experiments employ both model-based and deep RL implementations depending on the environment.

\subsection{Induced MDP}
\label{sec:induced-mdp}

In PSRO, the best response for player $i$ is obtained by \emph{reducing} the multi-agent stochastic game to a single-agent MDP: we fix the opponents’ meta-strategy profile, absorb their behavior into the environment, and solve the resulting induced MDP.

PSRO assumes that the game has \emph{perfect recall}, meaning that each agent remembers the entire sequence of states and actions it has observed and taken so far. 
This assumption ensures that no relevant information is forgotten when making decisions over time. 
Under perfect recall, Kuhn’s theorem~\cite{kuhn1953extensive} states that any mixed strategy (a distribution over full policies) is realization-equivalent to a behavior strategy that specifies, at each state, a distribution over available actions. 
In other words, while playing against a mixed strategy formally induces a POMDP—since the opponent’s policy index is unobserved—the assumption of perfect recall implies that the state $s$ already encapsulates the relevant history $h$. 
Hence, the two representations are equivalent, and without loss of generality the opponents’ meta-strategy $\sigma_{-i}$ can be represented as a behavior strategy $\sigma_{-i}(a_{-i}\mid s)$.\footnote{%
\emph{Remark:} If opponents’ meta-strategies are interpreted as distributions over full policies and one such policy is sampled at the start of an episode, then from the viewpoint of player $i$ the opponent policy index acts as a hidden variable, inducing a POMDP. Two standard resolutions exist: (i) represent the state space $\mathcal{S}$ as the set of complete histories, making the process Markov, or (ii) assume \emph{perfect recall}, which by Kuhn’s theorem guarantees that every mixed strategy is realization-equivalent to a behavior strategy. In both cases, we may treat the meta-strategy $\sigma_{-i}$ as a behavior strategy $\sigma_{-i}(a_{-i}\mid s)$, yielding the induced MDP.}

Formally, fixing $\sigma_{-i}$ yields effective transition and reward functions for player $i$ by marginalizing over the opponents’ actions:
\begin{align}
P^{\sigma_{-i}}(s'\mid s,a_i)
&= \sum_{a_{-i}\in\mathcal{A}_{-i}} \sigma_{-i}(a_{-i}\mid s)\; P(s'\mid s,a_i,a_{-i}), \label{eq:induced-dyn}\\
r_i^{\sigma_{-i}}(s,a_i)
&= \sum_{a_{-i}\in\mathcal{A}_{-i}} \sigma_{-i}(a_{-i}\mid s)\; r_i(s,a_i,a_{-i}). \label{eq:induced-rew}
\end{align}
(Replace sums by integrals for continuous $\mathcal{A}_{-i}$.) This defines the induced single-agent MDP for player $i$:
\[
\mathcal{M}_i^{\sigma_{-i}} \;=\; \big(\mathcal{S},\,\mathcal{A}_i,\,P^{\sigma_{-i}},\,r_i^{\sigma_{-i}},\,\gamma,\,\rho_0\big).
\]

\subsection{Independent Best Response}
\label{sec:ind-br}

Let $Q_i(s,a_i)$ denote the action-value function of player $i$ in the induced MDP $\mathcal{M}_i^{\sigma_{-i}}$. 
Value iteration computes $Q_i$ by repeatedly applying the Bellman optimality update:
\[
Q_i^{k+1}(s,a_i)
\;=\;
r_i^{\sigma_{-i}}(s,a_i)
\;+\;
\gamma \sum_{s'\in\mathcal{S}} P^{\sigma_{-i}}(s'\mid s,a_i)\,
\max_{a_i'\in\mathcal{A}_i} Q_i^{k}(s',a_i').
\]
As $k \to \infty$, the sequence converges to the optimal action-value function $Q_i^\star$.  
A greedy policy with respect to $Q_i^\star$,
\[
\pi_i^{\mathrm{BR}}(s) \;\in\; \arg\max_{a_i\in\mathcal{A}_i} Q_i^\star(s,a_i),
\]
is then a best response to $\sigma_{-i}$.

\subsection{Joint Experience Best Response}
\label{sec:jbr}

While Independent Best Response (\S\ref{sec:ind-br}) is the standard approach in PSRO, it has a key drawback: each agent’s best response is trained separately. 
This creates substantial redundancy: when training one agent, the interaction experience generated cannot be reused by others, leading to repeated and costly environment simulations. 
In domains with many agents or expensive dynamics, this per-agent training quickly becomes impractical.

To address this inefficiency, we propose \emph{Joint Experience Best Response (JBR)}. 
The core idea is to collect trajectories once, under the current meta-strategy profile $\sigma = (\sigma_j)_{j\in N}$, and reuse this joint dataset to compute best responses for all agents simultaneously. 
In this setting, all players act according to their meta-strategies during data collection, so the dataset $\mathcal{D}^\sigma$ consists of trajectories 
\[
\mathcal{D}^\sigma = \{ (s_t, a_{1,t},\dots,a_{n,t}, r_{1,t},\dots,r_{n,t}, s_{t+1}) \}_{t=0}^T
\]
generated by the Markov game when the joint policy profile $\sigma$ is executed. The overall procedure is summarized in Algorithm~\ref{alg:jbrpsro}.

\begin{algorithm}
\caption{PSRO with Joint Experience Best Response }
\label{alg:jbrpsro}
\begin{algorithmic}[1]
\REQUIRE Initial restricted policy sets $(X_i)$
\STATE Estimate $\hat{\mathcal{G}} = (N,(X_i),(\hat{u}_i))$
\STATE Initialize $\sigma \leftarrow \text{MSS}(\hat{\mathcal{G}})$
\WHILE{not terminated}
    \STATE Collect joint dataset $\mathcal{D}^{\sigma}$ under meta-strategy $\sigma$
    \FOR{player $i \in N$}
        \STATE Compute $\pi_i^{\mathrm{BR}} \leftarrow \text{RO}(\mathcal{D}_{\sigma_{-i}})$
        \STATE Update $X_i \leftarrow X_i \cup \{\pi_i^{\mathrm{BR}}\}$
    \ENDFOR
    \STATE Re-estimate $\hat{\mathcal{G}}$
    \STATE Update $\sigma \leftarrow \text{MSS}(\hat{\mathcal{G}})$
\ENDWHILE
\RETURN $\sigma$
\end{algorithmic}
\end{algorithm}

\subsection{Na\"ive JBR}
\label{sec:jbr-naive}

The most direct approach is to treat $\mathcal{D}^\sigma$ as a fixed offline dataset and then, for each agent $i$, construct its induced MDP $\mathcal{M}_i^{\sigma_{-i}}$ as in \S\ref{sec:induced-mdp}. 
Best responses are then obtained by applying model-based value iteration on $\mathcal{M}_i^{\sigma_{-i}}$, using $\mathcal{D}^\sigma$ to estimate the transition dynamics $P^{\sigma_{-i}}$ and rewards $r_i^{\sigma_{-i}}$. 
Formally, the $Q$-iteration update is
\[
Q_i^{k+1}(s,a_i)
\;=\;
\hat{r}_i^{\sigma_{-i}}(s,a_i)
\;+\;
\gamma \sum_{s'\in\mathcal{S}} \hat{P}^{\sigma_{-i}}(s'\mid s,a_i)\,
\max_{a_i'} Q_i^k(s',a_i'),
\]
where $\hat{r}_i^{\sigma_{-i}}$ and $\hat{P}^{\sigma_{-i}}$ are the reward and transition models estimated from $\mathcal{D}^\sigma$, respectively. 
The resulting greedy policy $\pi_i^{\mathrm{BR}}$ is then added to the restricted game, exactly as in standard PSRO.

Although conceptually simple, this \emph{Naïve JBR} suffers from a critical limitation: since the dataset $\mathcal{D}^\sigma$ is generated under the current meta-strategy, it may not cover the state–action regions relevant for computing the true best response. 
This distribution shift makes the problem an instance of offline RL, potentially leading to suboptimal approximations of best responses. 
To address these limitations, we introduce three complementary remedies. 
The first is a conservative safeguard based on Safe Policy Improvement (\S\ref{sec:jbr-safe}), which ensures that policies do not underperform the current meta-strategy. 
The second is an exploration-based approach (\S\ref{sec:jbr-exploration}), which perturbs the meta-strategy during data collection to improve coverage and admits theoretical guarantees. 
Finally, a hybrid variant (\S\ref{sec:hybrid-br}) alternates between joint and independent best responses to balance efficiency and accuracy. 
Together, these three approaches mitigate the shortcomings of Naïve JBR.

\subsection{Conservative JBR}
\label{sec:jbr-safe}

To address the offline RL bias of Naïve JBR, we adopt the principle of \emph{Safe Policy Improvement} (SPI)~\cite{laroche2019safe}. 
The idea is to restrict the learned policy to only deviate from its baseline $\sigma_i$ in regions of the state--action space that are sufficiently well supported by $\mathcal{D}^\sigma$. 
Wherever the dataset lacks coverage, the policy is forced to fall back to $\sigma_i$, guaranteeing that performance cannot degrade below the baseline.

Formally, let $N_{\mathcal{D}}(s,a_i)$ denote the number of times $(s,a_i)$ appears in $\mathcal{D}^\sigma$, and fix a threshold $N_\wedge$. 
We define the set of \emph{uncertain} state--action pairs as
\[
B_i \;=\; \{(s,a_i) \;\mid\; N_{\mathcal{D}}(s,a_i) < N_\wedge \}.
\]
The SPI constraint then requires that for all $(s,a_i)\in B_i$,
\[
\pi_i(a_i \mid s) = \sigma_i(a_i \mid s).
\]
That is, in poorly covered regions, the new policy exactly reproduces the baseline meta-strategy $\sigma_i$. 
Outside $B_i$, the policy is optimized via offline value iteration on $\mathcal{M}_i^{\sigma_{-i}}$, using $\mathcal{D}^\sigma$ to estimate transitions and rewards. 

The resulting policy $\pi_i^{\mathrm{SPI}}$ is not necessarily a best response to $\sigma_{-i}$, 
but it is guaranteed to be no worse than $\sigma_i$ (up to estimation error), 
while strictly improving wherever $\mathcal{D}^\sigma$ provides sufficient coverage. 
Thus, $\pi_i^{\mathrm{SPI}}$ can safely be added to the restricted policy set in PSRO. 
However, as we show in our experiments, this conservative approach, while outperforming Naïve JBR, 
still performs poorly compared to the original PSRO with independent best responses. 
This motivates the need for a more effective remedy.

\subsection{Exploration-Augmented JBR}
\label{sec:jbr-exploration}

While Conservative JBR prevents performance degradation, it remains overly cautious: 
improvements are confined to regions already well represented in the dataset, 
limiting its ability to match the performance of independent best responses. 
To overcome this, we propose to deliberately broaden coverage during data collection 
by perturbing the meta-strategy with exploration.

Let $\sigma = (\sigma_j)_{j\in N}$ be the current meta-strategy profile. 
During data collection we use an \emph{$\delta$-perturbed meta-strategy}  that keeps agents mostly on-policy w.r.t.\ $\sigma$ but injects exploration with probability $\delta$.
For each agent $j$, define a per-state mixed behavior
\begin{equation}
\tilde\sigma_j^{(\delta,\nu_j)}(a_j \mid s)
\;\;=\;\;
(1-\delta)\,\sigma_j(a_j \mid s)\;+\;\delta\,\nu_j(a_j \mid s),
\label{eq:eps-meta-single}
\end{equation}
where $\nu_j$ is an exploration policy. 
Agents act independently according to $\tilde\sigma^{(\delta,\nu)} = \big(\tilde\sigma_j^{(\delta,\nu_j)}\big)_{j\in N}$, generating a joint dataset $\mathcal{D}^{\tilde\sigma}$.

We consider two concrete choices, yielding two exploration strategies:

\subsubsection{$\delta$-random exploration.}
\label{sec:jbr-exploration-random}
Set $\nu_j^{\mathrm{rand}}(a_j\mid s)= \mathrm{Unif}(\mathcal{A}_j)$, i.e.,
\[
\tilde\sigma_j^{(\delta,\mathrm{rand})}(a_j\mid s)
=
(1-\delta)\,\sigma_j(a_j\mid s)\;+\;\delta\,\tfrac{1}{|\mathcal{A}_j|}.
\]
This injects state-agnostic coverage while keeping the process mostly on-policy w.r.t.\ $\sigma$.
\subsubsection{$\delta$-targeted exploration.}
\label{sec:jbr-exploration-targeted}
Let $\pi_{j}^{\mathrm{BR,cur}}$ denote the current BR candidate for agent $j$ (the policy being trained in JBR this iteration). 
Set $\nu_j^{\mathrm{tgt}}(a_j\mid s)=\pi_{j}^{\mathrm{BR,cur}}(a_j\mid s)$, i.e.,
\[
\tilde\sigma_j^{(\delta,\mathrm{tgt})}(a_j\mid s)
=
(1-\delta)\,\sigma_j(a_j\mid s)\;+\;\delta\,\pi_{j}^{\mathrm{BR,cur}}(a_j\mid s).
\]
This concentrates additional samples near currently promising BR behaviors, improving value estimation where it matters for BR computation.

\subsubsection{Theoretical Implication}
\label{sec:guarantees}

While exploration broadens coverage and mitigates offline RL bias, it also changes the learning target: 
agents no longer best respond to the exact meta-strategy $\sigma$, but rather to its $\delta$-perturbed variant $\tilde\sigma$. 
This raises the question of whether PSRO with exploration still preserves its theoretical guarantees. 

We now analyze this effect. In particular, we show that even though exploration changes the response oracle’s target distribution, 
PSRO with exploration-augmented JBR still preserves its convergence guarantees up to a bounded error. 
We focus on finite two-player zero-sum Markov games with bounded payoffs in $[\underline{u}, \overline{u}]$ 
and range $R=\overline{u}-\underline{u}$.

\begin{theorem}[Exploration-augmented JBR]
\label{thm:exploration-guarantee}
Let $\sigma$ be the current meta-strategy profile and $\tilde\sigma$ its $\delta$-perturbed variant used for data collection. If each agent computes an $\varepsilon$-best response to $\tilde\sigma$, then upon termination the resulting meta-strategy is an $(\varepsilon + 2R\delta)$-Nash equilibrium of the original game.
\end{theorem}

\paragraph{Proof sketch.}
Payoffs are linear in mixed strategies, so perturbing the opponents’ meta-strategy by $\delta$ changes the payoff of any policy by at most $2R\delta$. 
Hence a best response to $\tilde\sigma$ is also an approximate best response to $\sigma$, up to this error. 
Since PSRO’s termination check is performed under $\sigma$, no profitable deviations are missed, giving the stated $(\varepsilon+2R\delta)$ bound. 
A full proof is provided in Appendix~\ref{app:exploration-proof}.

\paragraph{Implication.}
This result shows that exploration-augmented JBR preserves the convergence guarantees of PSRO up to a controlled error: the deviation from equilibrium grows only linearly with $\delta$, while the added exploration improves coverage and reduces offline RL bias. 
It should be noted, however, that the guarantee holds only under the assumption that an $\varepsilon$-best response to $\tilde\sigma$ can be found. 
In practice, computing such a response from a fixed dataset is not guaranteed, as limited coverage may prevent reliable policy improvement. 
This motivates the use of multiple exploration strategies, which aim to enrich the dataset and increase the likelihood that a sufficiently good response can be learned in practice.

\subsection{Hybrid Best Response}
\label{sec:hybrid-br}

While JBR substantially reduces simulation cost, its reliance on offline data introduces approximation errors. 
Variants such as Conservative JBR and Exploration-Augmented JBR mitigate these errors, 
but they may still deviate from the accuracy of Independent BR (IBR), which computes exact best responses at each iteration. 
This raises the question of whether we can combine the efficiency of JBR with the accuracy of IBR. 

We therefore consider a \emph{Hybrid BR} approach that alternates between the two. 
Specifically, JBR (possibly in one of its variants) is used in most iterations, 
while every $k$-th iteration is replaced by a standard Independent BR update. 
Formally, let $t$ denote the iteration counter. 
If $t \bmod k = 0$, each agent $i$ computes a best response $\pi_i^{\mathrm{IBR}}$ via independent training against $\sigma_{-i}$; 
otherwise, responses are obtained using JBR (e.g., Naïve, Conservative, or Exploration-Augmented). 

This hybrid strategy retains the sample efficiency benefits of JBR in the majority of iterations, 
while periodically correcting approximation errors by injecting accurate best responses from IBR. 
In our experiments, we evaluate whether such periodic corrections are sufficient 
to close the performance gap between JBR-based variants and original PSRO.

\section{Experiments}
\label{sec:experiments}

The aim of our experiments is to assess the \emph{feasibility and reliability} of the proposed JBR framework. 
By design, JBR computes best responses for all agents simultaneously from a single shared dataset, which trivially improves sample efficiency by reducing environment interactions. 
The central question, however, is whether JBR can still produce accurate best responses—and thus preserve PSRO’s convergence behavior—when data are reused across agents. 

\subsection{Experimental Setup}

\paragraph{Games.}
We consider two classic two-player zero-sum poker games, \emph{Kuhn Poker}~\cite{kuhn2016simplified} and \emph{Leduc Poker}~\cite{southey2012bayes}, followed by three continuous multi-agent particle environments~\cite{lowe2017multi}: \emph{Simple Tag}, \emph{Simple Adversary}, and \emph{Simple Push}. 
This sequence increases the difficulty of learning reliable best responses: 
Kuhn Poker is small and fully covered by the meta-strategy, while Leduc Poker introduces imperfect information and a larger state space that exposes offline-learning bias. 
The particle environments further stress-test the method under continuous control and partial observability, where joint data reuse may degrade policy quality.

\paragraph{Methods.}
We evaluate the following algorithms:
\begin{itemize}
    \item \textbf{PSRO}: standard PSRO with independent best responses (baseline);
    \item \textbf{JBR-PSRO}: na\"ive joint-experience variant computing all BRs from shared data;
    \item \textbf{JBR-PSRO-SPI}: conservative JBR using Safe Policy Improvement (SPI). 
    The SPI variant requires a coverage threshold $N_{\wedge}$ to determine when to revert to the baseline meta-strategy. 
    To isolate the method’s intrinsic potential rather than its sensitivity to this hyperparameter, we employ an oracle-tuned setting that iterates over $N_{\wedge} \in [0, 50]$ at each iteration and selects the best-performing value. 
    While not practical, this allows us to evaluate the idealized upper bound of SPI performance;
    \item \textbf{JBR-PSRO-$\delta$R}: exploration-augmented JBR with $\delta$-random exploration;
    \item \textbf{JBR-PSRO-$\delta$T}: exploration-augmented JBR with $\delta$-targeted exploration;
    \item \textbf{HBR-PSRO($k$)}: hybrid variant alternating between JBR and IBR every $k$ iterations, 
    which can also be combined with exploration (e.g., \textbf{HBR-PSRO($k$)-$\delta$T}).
\end{itemize}
For comparison with standard MARL paradigms, we also include \textbf{IL} (independent learners using DDPG~\cite{lillicrap2015continuous}) and \textbf{CTDE} (centralized training with decentralized execution using MADDPG~\cite{lowe2017multi}).

\paragraph{Implementation details.}
All PSRO variants use the same meta-strategy solver, namely \emph{projected replicator dynamics}~\cite{lanctot2017unified}, to compute the Nash equilibrium of the empirical payoff matrix at each iteration. 
In the poker domains, best responses (BRs) are computed via model-based value iteration, while in the particle environments we employ deep RL oracles implemented with DDPG or MADDPG, depending on whether training is independent or centralized. 
Each iteration adds one new policy per agent and updates the meta-strategy accordingly. Although some particle environments involve more than two individual agents—such as \emph{Simple Tag} (one runner vs. two adversaries) and \emph{Simple Adversary} (two cooperating agents vs. one adversary)—we treat them as two-player games by grouping agents into two opposing teams: \emph{agents} and \emph{adversaries}. 
This team-based formulation follows the \emph{Team PSRO} framework~\cite{mcaleer2023team}, ensuring that equilibrium analysis remains well-defined at the team level. For best-response computation, we employ DDPG when the team consists of a single agent and MADDPG when the team includes multiple cooperating agents.

Following prior work~\cite{lanctot2017unified,zhou2022efficient}, we run PSRO for $100$ iterations in the poker games, allocating a budget of $10{,}000$ episodes for computing best responses. 
In the particle environments, we run $50$ iterations with a BR budget of $50{,}000$ episodes per iteration. 
Importantly, in standard PSRO each agent receives this full budget independently, whereas in JBR all agents \emph{share} the same budget to compute their BRs jointly, reflecting the gain in training efficiency. 
Unless otherwise stated, the exploration rate $\delta$ is set to $0.1$ for \textbf{JBR-PSRO-$\delta$R} and $0.5$ for \textbf{JBR-PSRO-$\delta$T}; sensitivity analyses for these values are reported in Section~\ref{sec:ablation_delta}. 
All reported results are averaged over three random seeds.

To evaluate convergence, we track two metrics: (i) \emph{NashConv}, measuring equilibrium accuracy, and (ii) \emph{total BR episodes}, quantifying overall sample efficiency. 
In small discrete games such as Kuhn and Leduc Poker, NashConv is computed exactly using OpenSpiel’s implementation~\cite{lanctot2019openspiel}, where a full tree search is performed to obtain the true best response. 
In continuous particle environments, where exact BRs are intractable, we instead estimate an \emph{approximate NashConv}: for each agent we train three BR candidates with different random initializations and an extended BR budget of $100{,}000$ episodes, and select the best-performing one to approximate the true best response.

\begin{figure}[t]
    \centering
    \includegraphics[width=1\columnwidth]{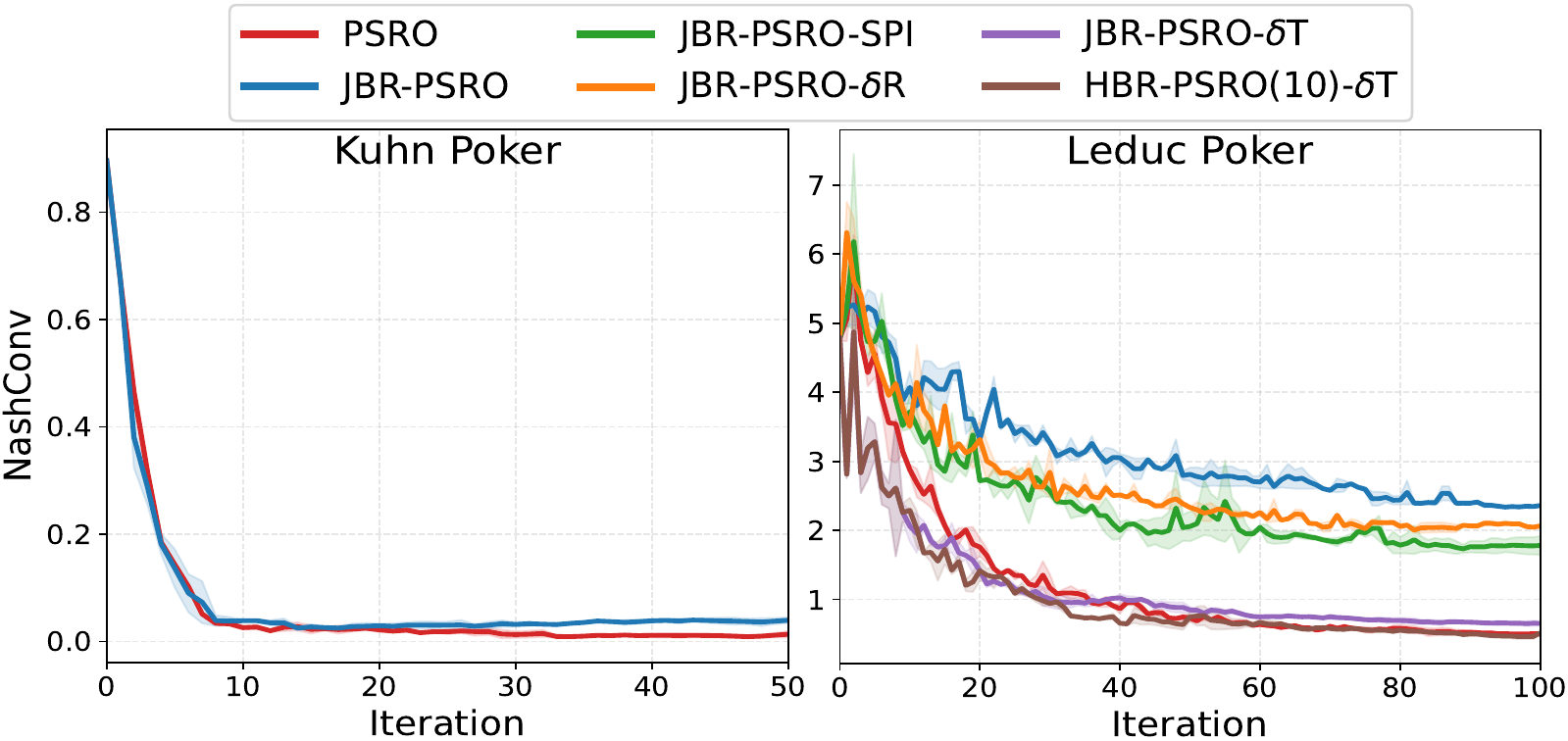}
    \caption{
    Convergence of PSRO and JBR in two poker games of increasing complexity. 
    \textbf{Left:} In \emph{Kuhn Poker}, JBR remains close to PSRO, indicating that it is feasible when the state space is small and well covered. 
    \textbf{Right:} In \emph{Leduc Poker}, naive JBR diverges from PSRO as complexity grows, revealing the offline-learning bias that motivates the JBR variants.
    }
    \Description{
    Convergence of PSRO and JBR in two poker games of increasing complexity. 
    \textbf{Left:} In \emph{Kuhn Poker}, JBR remains close to PSRO, indicating that it is feasible when the state space is small and well covered. 
    \textbf{Right:} In \emph{Leduc Poker}, naive JBR diverges from PSRO as complexity grows, revealing the offline-learning bias that motivates the JBR variants.
    }
    \label{fig:kuhn_leduc}
\end{figure}

\subsection{Is Joint Experience Best Response Feasible?}

We first investigate JBR can serve as a reliable alternative to independent best responses.
We evaluate \textbf{PSRO} and \textbf{JBR-PSRO} in \emph{Kuhn Poker}, a minimal imperfect-information game with a small and fully covered state space.

As shown in the left panel of Figure~\ref{fig:kuhn_leduc}, the performance gap between JBR-PSRO and PSRO remains small across iterations, indicating that JBR closely follows the convergence behavior of standard PSRO. 
Because Kuhn Poker’s state space is compact and every state–action pair is visited under the joint meta-strategy, there is no distribution-shift bias when computing best responses from shared data. 
This confirms that JBR is \emph{feasible} in low-complexity domains where the induced MDPs are well covered by on-policy trajectories.

\subsection{When Does Na\"ive JBR Fail?}

We next examine how JBR behaves as best-response computation becomes more complex. 
To this end, we compare \textbf{PSRO} and \textbf{JBR-PSRO} in \emph{Leduc Poker}, which features a significantly larger state space and deeper decision structure. 
Unlike Kuhn Poker, the joint dataset collected under the meta-strategy covers only a subset of reachable trajectories, leading to an offline-learning bias when best responses are computed from shared experience.

As shown in the right panel of Figure~\ref{fig:kuhn_leduc}, the performance gap between JBR-PSRO and PSRO increases considerably: while PSRO continues to reduce NashConv, JBR-PSRO stagnates at higher exploitability. 
This demonstrates that na\"ive JBR, although efficient, fails to maintain equilibrium accuracy once the induced MDPs become partially observable or high-dimensional. 
These results motivate the development of the \emph{JBR variants} discussed next.

\subsection{How Effective Are the JBR Variants?}

We now evaluate whether the enhanced \emph{JBR variants} can mitigate the offline-learning bias observed in na\"ive JBR. 
Our focus is on the performance of \textbf{JBR-PSRO-SPI}, \textbf{JBR-PSRO-$\delta$R}, and \textbf{JBR-PSRO-$\delta$T} compared to the baseline \textbf{JBR-PSRO}. 

As shown in the right panel of Figure~\ref{fig:kuhn_leduc}, all enhanced JBR variants outperform the na\"ive baseline, confirming that explicit mechanisms for coverage and stability substantially improve best-response quality. 
However, their relative effectiveness differs: 
\textbf{JBR-PSRO-SPI} is conservative, as its Safe Policy Improvement constraint forces the policy to revert to the baseline meta-strategy in under-sampled regions, limiting progress. 
\textbf{JBR-PSRO-$\delta$R}, which employs random exploration, increases data diversity but lacks direction, leading to only moderate gains. 
In contrast, \textbf{JBR-PSRO-$\delta$T} achieves near-PSRO performance by guiding exploration toward promising trajectories. 
Its $\delta$-targeted exploration explicitly conditions on the current best response during data collection—anticipating how opponents' strategies will change after the next update. 
This forward-looking behavior resembles \emph{opponent modeling}, enabling more accurate and stable best responses. 
Interestingly, this forecasting nature also explains the faster convergence observed in early iterations of \textbf{JBR-PSRO-$\delta$T} compared to standard PSRO, as agents proactively adapt to the direction of the upcoming meta-strategy updates.

Based on these observations, we identify \textbf{JBR-PSRO-$\delta$T} as the most effective variant. 
To further close the small remaining gap to PSRO, we test the hybrid approach that periodically replaces JBR updates with independent best responses. 
As shown in Figure~\ref{fig:kuhn_leduc}, this configuration, \textbf{HBR-PSRO(10)-$\delta$T}, achieves PSRO-level convergence while preserving the substantial sample-efficiency gains of JBR.

\subsection{How Sample-Efficient are the JBR variants?}

We now analyze the \emph{sample efficiency of best-response computation} across different PSRO methods. 
Figure~\ref{fig:efficiency_tradeoff} reports results in \emph{Leduc Poker}, showing the total number of environment episodes used for best-response training (in millions, $x$-axis) versus the minimum NashConv achieved after 100 iterations ($y$-axis). 
This directly quantifies how efficiently each method computes best responses relative to its equilibrium quality.

Standard \textbf{PSRO} achieves the lowest NashConv but is highly sample-inefficient, as each agent independently trains a new best response at every iteration. 
In contrast, \textbf{JBR-PSRO} and its enhanced variants reuse a single shared dataset to compute best responses for all agents simultaneously, effectively amortizing the simulation cost across agents. 
This substantially reduces the number of environment interactions required for best-response computation while maintaining competitive equilibrium accuracy.

Among the enhanced variants, \textbf{JBR-PSRO-$\delta$T} achieves the best trade-off between equilibrium accuracy and BR sample efficiency, reaching near-PSRO NashConv with half of PSRO’s total BR episodes. 
Its hybrid counterparts, \textbf{HBR-PSRO(10)-$\delta$T} and \textbf{HBR-PSRO(30)-$\delta$T}, further close the remaining gap to PSRO by occasionally performing independent BR updates, requiring only a small increase in BR episodes. 
\textbf{JBR-PSRO-SPI} remains more stable but conservative, while \textbf{JBR-PSRO-$\delta$R} provides moderate improvements due to uninformed exploration. 
Overall, these results confirm that JBR-based training substantially improves the \emph{sample efficiency of best-response computation}, enabling practical deployment of PSRO in domains where environment interaction cost dominates runtime.

\begin{figure}[t]
    \centering
    \includegraphics[width=1\linewidth]{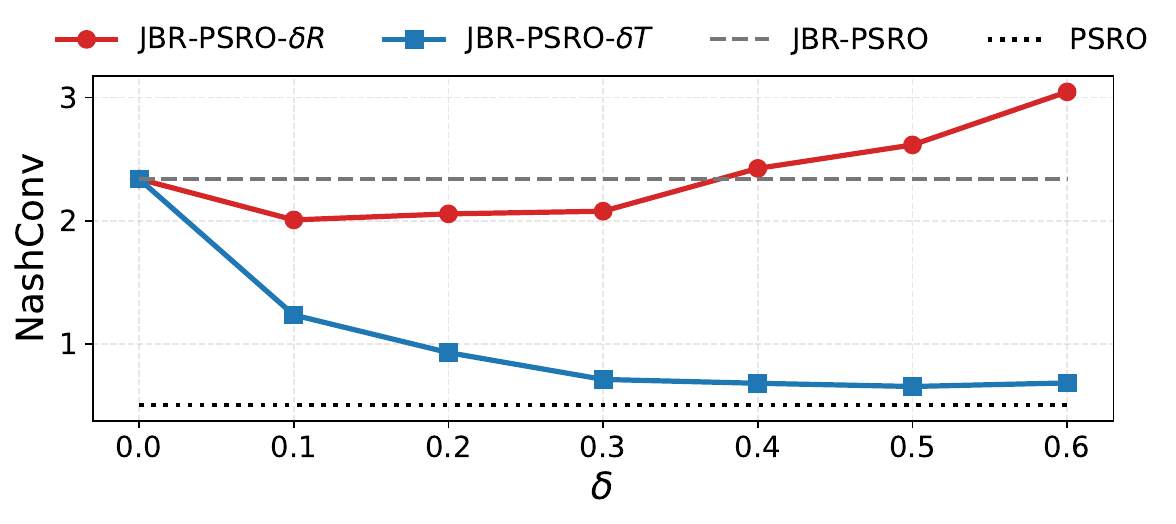}
    \caption{
    \textbf{Effect of the exploration rate $\delta$ in Leduc Poker.} 
    Minimum NashConv after 100 iterations for random 
    (\textsc{JBR-PSRO-$\delta$R}) and targeted 
    (\textsc{JBR-PSRO-$\delta$T}) exploration. 
    Random exploration peaks at $\delta{=}0.1$ then degrades beyond $0.4$, 
    while targeted exploration improves up to $\delta{=}0.5$ and remains 
    consistently better than naïve JBR.
    }
    \Description{
    \textbf{Effect of the exploration rate $\delta$ in Leduc Poker.} 
    Minimum NashConv after 100 iterations for random 
    (\textsc{JBR-PSRO-$\delta$R}) and targeted 
    (\textsc{JBR-PSRO-$\delta$T}) exploration. 
    Random exploration peaks at $\delta{=}0.1$ then degrades beyond $0.4$, 
    while targeted exploration improves up to $\delta{=}0.5$ and remains 
    consistently better than naïve JBR.
    }
    \label{fig:delta_ablation}
\end{figure}

\begin{figure}[t]
    \centering
    \includegraphics[width=1\columnwidth]{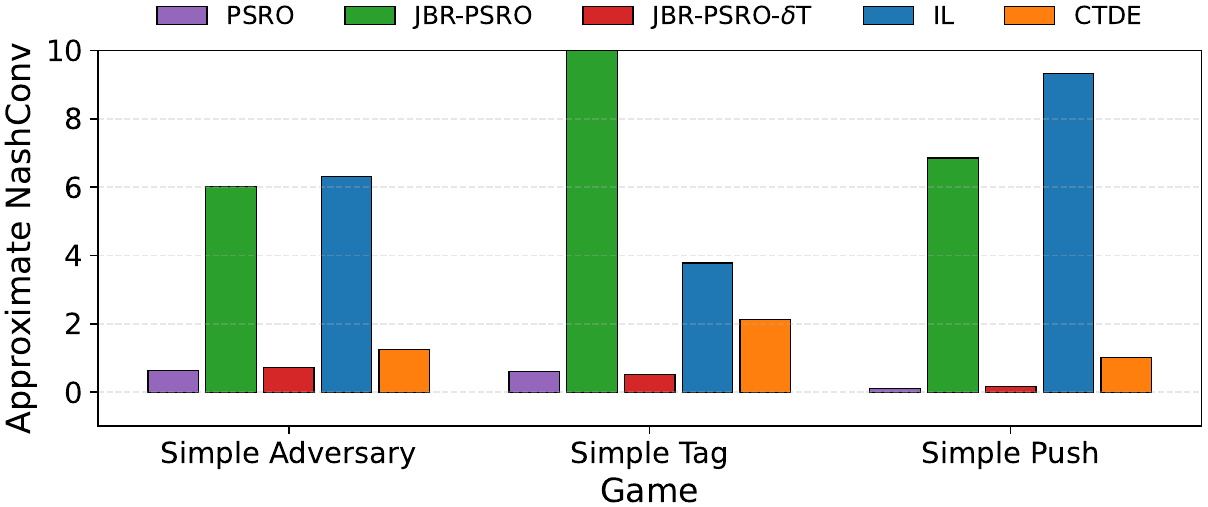}
    \caption{
    \textbf{Approximate NashConv in continuous multi-agent environments.}
    Comparison of \textsc{PSRO}, \textsc{JBR-PSRO}, \textsc{JBR-PSRO-$\delta$T}, and two MARL baselines (\textsc{IL/DDPG}, \textsc{CTDE/MADDPG}) on \emph{Simple Tag}, \emph{Simple Adversary}, and \emph{Simple Push}. 
    \textsc{PSRO} and \textsc{JBR-PSRO-$\delta$T} achieve comparable approximate NashConv across games; both PSRO methods outperform \textsc{IL} and \textsc{CTDE}.
    }
    \Description{
    \textbf{Approximate NashConv in continuous multi-agent environments.}
    Comparison of \textsc{PSRO}, \textsc{JBR-PSRO}, \textsc{JBR-PSRO-$\delta$T}, and two MARL baselines (\textsc{IL/DDPG}, \textsc{CTDE/MADDPG}) on \emph{Simple Tag}, \emph{Simple Adversary}, and \emph{Simple Push}. 
    \textsc{PSRO} and \textsc{JBR-PSRO-$\delta$T} achieve comparable approximate NashConv across games; both PSRO methods outperform \textsc{IL} and \textsc{CTDE}.
    }
    \label{fig:particles_min_nashconv}
\end{figure}

\subsection{Sensitivity to the Exploration Rate \texorpdfstring{$\boldsymbol{\delta}$}{delta}}
\label{sec:ablation_delta}

We next study how $\delta$ influences performance in \emph{Leduc Poker}. 
Figure~\ref{fig:delta_ablation} reports the final NashConv after 100 iterations for both \textbf{JBR-PSRO-$\delta$R} and \textbf{JBR-PSRO-$\delta$T} across different $\delta$ values.

For \textbf{JBR-PSRO-$\delta$R}, performance initially improves as $\delta$ increases, reaching its best value at $\delta = 0.1$. 
Beyond this point, higher levels of random mixing gradually degrade performance, and after $\delta = 0.4$, results become worse than the naïve \textbf{JBR-PSRO} baseline. 
This behavior reflects the instability introduced by excessive randomization, which distorts value estimates and weakens the best-response signal. In contrast, \textbf{JBR-PSRO-$\delta$T} shows a steady improvement as $\delta$ increases, plateauing around $\delta = 0.4$ and peaking near $\delta = 0.5$. 
A slight drop occurs at $\delta = 0.6$, but performance remains well above that of \textbf{JBR-PSRO} for all tested values. 
This robustness indicates that targeted exploration continues to enhance coverage and opponent forecasting over a broad range of $\delta$.

Overall, these results confirm that random and targeted exploration behave differently: 
random exploration requires a small, cautious perturbation to avoid destabilizing learning, whereas targeted exploration benefits from stronger mixing that accentuates the predictive structure of the meta-strategy.

\subsection{Beyond Poker: Continuous Multi-Agent Environments}

We now evaluate whether the conclusions from poker generalize to continuous multi-agent control. 
In the \emph{Simple Tag}, \emph{Simple Adversary}, and \emph{Simple Push} particle environments, exact best responses are intractable; accordingly, we report an \emph{approximate NashConv}. 
For each agent, we train three BR candidates with different random initializations under an extended BR budget of $100{,}000$ episodes and select the best-performing candidate to approximate the true BR (the same protocol is used to score IL and CTDE).

Figure~\ref{fig:particles_min_nashconv} compares \textbf{PSRO}, \textbf{JBR-PSRO}, and \textbf{JBR-PSRO-$\delta$T} (the best-performing enhanced variant), alongside two standard MARL frameworks: \textbf{IL (DDPG)} and \textbf{CTDE (MADDPG)}. 
Across all three environments, \textbf{PSRO} and \textbf{JBR-PSRO-$\delta$T} attain comparable approximate NashConv, while \textbf{JBR-PSRO} lags behind—mirroring the Leduc findings and underscoring the need for targeted exploration when data are reused. 
Moreover, both PSRO and JBR-PSRO-$\delta$T consistently outperform \textbf{IL} and \textbf{CTDE}, indicating greater robustness to strategic non-stationarity. 
These results reinforce the usefulness of PSRO-style game-theoretic training in continuous domains and highlight why improving the sample efficiency of best-response computation is practically important.

\subsection{Summary of Empirical Insights}
\label{sec:summary_empirical}

Across all experiments, our results collectively demonstrate that JBR substantially improves the \emph{sample efficiency of best-response computation} in PSRO while preserving strong equilibrium performance when paired with targeted exploration. 
In small, fully covered games (\emph{Kuhn Poker}), JBR matches PSRO’s convergence behavior, establishing feasibility under well-sampled conditions. 
In larger games (\emph{Leduc Poker}), naïve JBR underperforms due to offline bias, but enhanced variants—especially \textbf{JBR-PSRO-$\delta$T}—restore accuracy by anticipating opponent updates; \textbf{HBR-PSRO($k$)-$\delta$T} closes the remaining gap to PSRO with only a modest increase in BR episodes. 
Ablations show that random exploration requires small mixing, whereas targeted exploration benefits from stronger perturbation (optimal around $\delta{=}0.5$). 
In continuous multi-agent control, \textbf{JBR-PSRO-$\delta$T} again achieves performance comparable to PSRO while being far more sample-efficient. 
Moreover, both PSRO and JBR-PSRO-$\delta$T demonstrate greater robustness than standard MARL frameworks such as IL and CTDE, highlighting the importance of PSRO-style game-theoretic training for stable coordination—and showing that improving the efficiency of best-response computation makes such training practically viable in larger-scale domains.

\section{Discussion}
\label{sec:discussion}

\paragraph{Practical implications.}
PSRO remains one of the most reliable approaches for learning strategically robust policies, yet its use is often constrained by the high cost of best-response computation. 
The proposed JBR framework mitigates this limitation by reusing experience across agents, with \textbf{JBR-PSRO-$\delta$T} providing the best trade-off between equilibrium accuracy and BR sample efficiency. 
When near-exact PSRO accuracy is required, hybrid updates (\textbf{HBR-PSRO($k$)-$\delta$T}) offer a simple way to recover full performance with minimal additional cost.

\paragraph{Interpretation.}
Targeted exploration acts as an implicit form of opponent modeling: by conditioning data collection on the current best response, agents forecast likely opponent adaptations, explaining the faster early convergence of \textbf{JBR-PSRO-$\delta$T}. 
This connection highlights how predictive exploration can improve both efficiency and stability in meta-strategy learning.

\paragraph{Limitations.}
Despite the gains, PSRO and JBR variants still require many iterations and substantially more samples than typical MARL methods. 
Scalability to larger policy populations and high-dimensional continuous domains remains challenging, as meta-strategy updates and payoff matrix growth introduce additional computational cost. 
Developing adaptive schedules for $\delta$ and $k$, or integrating learned priors to guide exploration, are promising directions for reducing this overhead.

\section{Conclusion}
\label{sec:conclusion}

We introduced the Joint Experience Best Response (JBR) framework, which reduces the cost of best-response computation in PSRO by reusing shared interaction data across agents. 
While naïve JBR can suffer from offline-learning bias in complex games, the enhanced variant \textbf{JBR-PSRO-$\delta$T} and its hybrid extension \textbf{HBR-PSRO($k$)-$\delta$T} recover PSRO-level equilibrium accuracy with far greater sample efficiency. 
Experiments across both discrete and continuous multi-agent domains show that JBR-based training maintains the robustness of PSRO while making it more practical for large-scale applications. 
Future work will explore adaptive exploration and hybrid schedules, scalable meta-strategy updates, and extensions of targeted exploration to deep continuous settings.

\appendix
\section{Proof of Theorem~\ref{thm:exploration-guarantee}}
\label{app:exploration-proof}

We provide the full proof that PSRO with exploration-augmented JBR preserves convergence guarantees up to a bounded error. We focus on finite two-player zero-sum Markov games with bounded payoffs in $[\underline{u}, \overline{u}]$ 
and range $R=\overline{u}-\underline{u}$.

\begin{lemma}[Linearity in mixtures]\label{lem:lin}
For fixed $\pi_i$, $u_i(\pi_i,\cdot)$ is linear in the opponent mixture, and 
$u_i(\cdot,\cdot)$ is bilinear in $(\sigma_i,\sigma_{-i})$.
\end{lemma}

\begin{proof}
Immediate from

\(
u_i(\sigma_i,\sigma_{-i})
=\sum_{\pi_i,\pi_{-i}}\sigma_i(\pi_i)\sigma_{-i}(\pi_{-i})u_i(\pi_i,\pi_{-i}).
\)
\end{proof}

\begin{lemma}[Stability under perturbations]\label{lem:stability}
For any $\pi_i$ and mixtures $\sigma_{-i},\tilde\sigma_{-i}$,
\[
\big|u_i(\pi_i,\sigma_{-i})-u_i(\pi_i,\tilde\sigma_{-i})\big|
\;\le\;\tfrac{R}{2}\,\|\sigma_{-i}-\tilde\sigma_{-i}\|_1.
\]
\end{lemma}

\begin{proof}
By Lemma~\ref{lem:lin}, $u_i(\pi_i,\cdot)$ is an expectation of a bounded function.  
The difference of expectations is bounded by $(R/2)\|\cdot\|_1$ (duality of total variation and bounded functions).
\end{proof}

\begin{lemma}[BR to perturbed $\Rightarrow$ approx-BR to true]\label{lem:abr}
Let $\pi_i^\star\in\bri(\sigma_{-i})$ and $\hat\pi_i$ be an $\varepsilon$-BR to $\tilde\sigma_{-i}$. Then
\[
u_i(\pi_i^\star,\sigma_{-i}) - u_i(\hat\pi_i,\sigma_{-i})
\;\le\;\varepsilon + R\,\|\sigma_{-i}-\tilde\sigma_{-i}\|_1.
\]
\end{lemma}

\begin{proof}
Adding and subtracting expectations under $\tilde\sigma_{-i}$ gives
\begin{multline*}
u_i(\pi_i^\star,\sigma_{-i}) - u_i(\hat\pi_i,\sigma_{-i}) \\
=\big[u_i(\pi_i^\star,\sigma_{-i}) - u_i(\pi_i^\star,\tilde\sigma_{-i})\big]
+ \big[u_i(\pi_i^\star,\tilde\sigma_{-i}) - u_i(\hat\pi_i,\tilde\sigma_{-i})\big] \\
+ \big[u_i(\hat\pi_i,\tilde\sigma_{-i}) - u_i(\hat\pi_i,\sigma_{-i})\big].
\end{multline*}
The middle term is $\le\varepsilon$ by optimality of $\hat\pi_i$ under $\tilde\sigma_{-i}$.  
The other two terms are each $\le(R/2)\|\sigma_{-i}-\tilde\sigma_{-i}\|_1$ by Lemma~\ref{lem:stability}. 
\end{proof}

\begin{theorem}[PSRO with perturbed targets]\label{thm:perturbed}
Suppose at termination each oracle computes an $\varepsilon$-BR to $\tilde\sigma_{-i}$ with 
$\|\tilde\sigma_{-i}-\sigma_{-i}\|_1 \le \Delta$.  
If the termination check is performed under $\sigma$, then $\sigma$ is an 
$(\varepsilon + 2R\Delta)$-Nash equilibrium.
\end{theorem}

\begin{proof}
Fix player $i$.  
Let $\pi_i^\star\in\bri(\sigma_{-i})$ and $\hat\pi_i$ the oracle’s $\varepsilon$-BR to $\tilde\sigma_{-i}$.  
By Lemma~\ref{lem:abr},
\[
u_i(\pi_i^\star,\sigma_{-i}) - u_i(\hat\pi_i,\sigma_{-i})
\;\le\;\varepsilon + R\Delta.
\]
By the termination condition,
$u_i(\sigma_i,\sigma_{-i}) \ge u_i(\hat\pi_i,\sigma_{-i})$.  
Thus
\[
u_i(\sigma_i,\sigma_{-i})
\;\ge\;\max_{\pi_i}u_i(\pi_i,\sigma_{-i}) - (\varepsilon+R\Delta),
\]
which is the required equilibrium condition. Applying the same argument to the other player yields the claim.
\end{proof}

During data collection we use per-state $\delta$-mixing:
\[
\tilde\sigma_j(a_j\!\mid\! s)=(1-\delta)\sigma_j(a_j\!\mid\! s)+\delta\nu_j(a_j\!\mid\! s),
\]
which implies
\(\|\tilde\sigma_j(\cdot\!\mid\! s)-\sigma_j(\cdot\!\mid\! s)\|_1 \le 2\delta\).
Hence the induced deviation parameter satisfies $\Delta\le 2\delta$.

\begin{corollary}[Exploration-augmented JBR]\label{cor:ea-jbr}
If each oracle computes an $\varepsilon$-BR to the $\delta$-mixed behavior, 
then the terminal meta-strategy of PSRO is an $(\varepsilon+2R\delta)$-Nash equilibrium.
\end{corollary}

This corollary establishes Theorem~\ref{thm:exploration-guarantee} from the main text.
\paragraph{Remark.}
The assumption that each agent can compute an $\varepsilon$-best response to 
$\tilde\sigma_{-i}$ should be understood in the same oracle sense as in standard PSRO: 
it abstracts away the mechanism used to obtain the best response. 
In practice, responses are trained on data collected under $\tilde\sigma$. 
Without perturbation ($\delta=0$), this reduces to an offline RL problem with 
potential support mismatch. With $\delta>0$, exploration ensures that all 
actions receive nonzero probability, so the dataset has broader coverage and 
the offline problem is better posed. Thus perturbation does not make the setting 
equivalent to online learning, but it guarantees that the $\varepsilon$-BR target 
is well defined and approximable from data. Related perspectives on enforcing equilibrium conditions from fixed multi-agent datasets—without online interaction—have also been explored in recent work on inverse and simulacral multi-agent learning~\cite{goktasefficient}.



\bibliographystyle{ACM-Reference-Format} 
\bibliography{sample}

\end{document}